\documentclass[12pt]{article}
\usepackage{amsmath,amssymb,fullpage}

\usepackage{enumerate}

\usepackage{graphicx}

\usepackage{makeidx}
\usepackage[utf8]{inputenc}
\usepackage{wrapfig}
\usepackage{amsmath}
\usepackage{amsfonts}
\usepackage{mathrsfs}
\usepackage{amssymb}
\usepackage{latexsym}
\usepackage{amsbsy}
\usepackage{color}
\usepackage{bbm}

\usepackage{mathrsfs}

\usepackage{wasysym}

\providecommand{\otherindexspace}[1]{}

\usepackage{amsthm}

\newtheorem{theorem}{Theorem}[section]

\newtheorem{lemma}[theorem]{Lemma}
\newtheorem{proposition}[theorem]{Proposition}
\newtheorem{remark}[theorem]{Remark}

\newtheorem{definition}[theorem]{Definition}

\newtheorem{example}[theorem]{Example}
\newtheorem{assumption}[theorem]{Assumption}

\numberwithin{equation}{section}

\def\cal#1{\mathcal{#1}}

\def \R{\mathbb {R}}

\usepackage{fancyhdr}

\usepackage{graphics}
\usepackage{graphicx}

\DeclareGraphicsExtensions{.eps,.bmp,.jpg,.pdf,.mps,.png,.gif}

\makeatletter
\def\titre{\@title}
\makeatother

\title{American options in an imperfect market with default}

\author{Roxana Dumitrescu\thanks{Department of Mathematics, King's College London, Strand, London, WC2R 2LS, United Kingdom, email: \textbf{roxana.dumitrescu@kcl.ac.uk}} \and Marie-Claire Quenez \thanks{LPMA,
Université Paris 7 Denis Diderot, Boite courrier 7012, 75251 Paris cedex 05, France, email: \textbf{quenez@math.univ-paris-diderot.fr}} \and  Agnès Sulem
\thanks{ INRIA Paris,  3 rue Simone Iff, CS 42112, 75589 Paris Cedex 12, France, email: \textbf{agnes.sulem@inria.fr}}}

\begin{document}

\date{\today}

\maketitle

\begin{abstract}

We study  pricing  and  (super)hedging  for American options in an imperfect market model 
with default, where the imperfections are taken into account via the nonlinearity  of the wealth dynamics. The payoff is given by an RCLL adapted process $(\xi_t)$. 
We  define the {\em seller's superhedging price} of the American option as the minimum of the initial capitals which allow the seller to build up a superhedging portfolio.
We {prove} that this price coincides with the value function of an optimal stopping problem with nonlinear expectations   induced by BSDEs with default jump, which corresponds to the solution of a  reflected BSDE with lower barrier. Moreover, we show the
 existence of a superhedging portfolio strategy. 
 We then consider the  {\em buyer's superhedging price}, which is defined as the supremum of the initial wealths which allow the buyer to select an exercise time $\tau$ and a portfolio strategy $\varphi$  so that he/she is superhedged.   Under the additional assumption of left upper semicontinuity along stopping times of $(\xi_t)$, we show the existence of a superhedge $(\tau, \varphi)$ for the buyer, as well as a  characterization of the buyer's superhedging price via the solution of a nonlinear reflected BSDE with upper barrier.

 \end{abstract}



\section{Introduction}


 We consider an American option associated with a terminal time $T$ and  a payoff given by an RCLL adapted  process $(\xi_t)$.
 The case of a classical perfect market has been largely studied in the literature (see e.g. \cite{KS,MR}). Recall that the {\em seller's superhedging price} (called also fair price in the literature), denoted by $u_0$, is classically defined as the minimal initial capital which enables the seller  to invest in a portfolio which  covers  his liability to pay to the buyer up to $T$ no matter 
what the exercise time chosen by the buyer.
Moreover, this price  is equal to the value function of the following optimal stopping time problem
\begin{equation}\label{Kifer}
 \sup_{\tau \in \mathcal{T}} \mathbb{E}_{Q}({\tilde \xi}_\tau),
\end{equation}
where $\mathcal{T}$ is a set of stopping times valued in  $[0,T]$. Here, $\tilde \xi_t$ denotes the discounted value of $\xi_t$,  equal to $e^{-rt} \xi_t$  in the Black and Scholes model, where $r$ is the instantaneous interest rate.
Moreover, $\mathbb{E}_{Q}$ denotes the expectation under the unique martingale probability measure $Q$ of the market model. In \cite{EPardQ}, the seller's superhedging price is characterized via a reflected BSDE with lower obstacle. 

The aim of the present paper is to study pricing and hedging issues for American options in 
 the case of  imperfections in the market model taken into account via the nonlinearity  of the wealth dynamics, modeled via a nonlinear driver $g$. We moreover include the possibility of a default. A large class of imperfect market models can fit in our framework, like different borrowing and lending interest rates, and the case when the seller of the option is a 
 "large trader" whose  hedging  strategy may affect the market prices and even the default probability. 
 
 
 We provide a  characterization of
 the seller's superhedging price $u_0$ as the value of a corresponding {\em optimal stopping problem} with a nonlinear expectation, more precisely
\begin{equation}\label{GD}
u_0= \sup_{\tau \in \mathcal{T}} \mathcal{E}^g(\xi_\tau), 
\end{equation}
where $\mathcal{E}^g$ is a nonlinear expectation/$g$-evaluation induced by a nonlinear BSDE with default jump solved under the primitive probability measure $P$ with driver $g$.
  Note that in the particular case of a perfect market, the driver $g$ is linear and \eqref{GD} reduces to $\eqref{Kifer}$. 
 We also show that the seller's superhedging price can be characterized via the solution of the reflected BSDE with driver $g$ and lower obstacle $(\xi_t)$, as well as the existence of a superhedging portfolio strategy for the seller.

We then consider the  {\em buyer's superhedging price}, denoted by $v_0$, which is defined as the supremum of the initial wealths which allow the buyer to select an exercise time $\tau$ and a portfolio strategy $\varphi$  so that he/she is superhedged.   Under the additional assumption of left upper semicontinuity along stopping times of $(\xi_t)$, we show the existence of a superhedge $(\tau, \varphi)$ for the buyer, as well as a  characterization of the buyer's superhedging price via the solution of a nonlinear reflected BSDE with upper barrier $(-\xi_t)$, that is
\begin{equation*}
v_0=
-\inf_{\nu \in {\mathcal T}}{\cal E}_{0,\nu}^g ( -\xi_{\nu}).
\end{equation*}
Note that in the classical case of a perfect market, the buyer's superhedging price is equal to the seller's superhedging price, that is $v_0=u_0$, since, in this case, the $g$-evaluation $\mathcal{E}^g$ is linear.

When $-g(t,-y,-z,-k) \leq g(t,y,z,k),$ then $v_0 \leq u_0$. The interval $[v_0,u_0]$ can then be seen as a non-arbitrage interval for the price of the American option in the sense of \cite{KaKou}. In the example of a higher interest rate for borrowing, this result corresponds to the one shown in \cite{KaKou} by a dual approach.

The paper is organized as follows:  in Section \ref{sec2}, we introduce our imperfect market model with default and nonlinear wealth dynamics. In Section \ref{SECAM}, we study  pricing and (super)hedging of American options from the seller's point of view. In Section  \ref{SECB}, we
address the pricing and (super)hedging problem from the buyer's point of view.

\section{Imperfect market model with default}\label{sec2}

\subsection{Market model with default}\label{marketmodel}
Let $(\Omega, \mathcal{G}, {P})$ be a complete probability space 
 equipped with two stochastic processes:
  a unidimensional standard Brownian motion $W$ and a jump process $N$ defined by 
  $N_t={\bf 1}_{\vartheta\leq t}$ for any $t\in[0,T]$, where $\vartheta$ is a random variable which models a default time. We assume that this default can appear at any time that is $P(\vartheta \geq t)>0$ for any $t\geq 0$. We denote by ${\mathbb G}=\{\mathcal{G}_t, t\geq 0 \}$ the {\em augmented filtration} that is generated by $W$ and $N$ (in the sense of \cite[IV-48]{DM1}). We suppose that  $W$ is a ${\mathbb G}$-Brownian motion. 
 We denote by ${\cal P}$ the ${\mathbb G}$-predictable $\sigma$-algebra.
 Let  $(\Lambda_t)$ be the  predictable compensator of the nondecreasing process $(N_t)$.
 Note that $(\Lambda_{t \wedge \vartheta})$ is then the predictable compensator of
  $(N_{t \wedge \vartheta} )= (N_t)$. By uniqueness of the predictable compensator, 
  $\Lambda_{t \wedge \vartheta} = \Lambda_t$, $t\geq0$ a.s.
  We assume that $\Lambda$ is absolutely continuous w.r.t. Lebesgue's measure, so that there exists a nonnegative process $\lambda$, 
 called the intensity process, such that $\Lambda_t=\int_0^t \lambda_s ds$, $t\geq0$.
  Since $\Lambda_{t \wedge \vartheta} = \Lambda_t$,  $\lambda$ vanishes after $\vartheta$. 
The compensated martingale   is given by
\begin{equation*}
M_t:= N_t-\int_0^t\lambda_sds\,.
\end{equation*}

Let $T >0$ be the terminal time. We define the following sets:
\begin{itemize}
\item ${S}^{2}$ 
is the set of ${\mathbb G}$-adapted RCLL processes $\varphi$ such that $\mathbb{E}[\sup_{0\leq t \leq T} |\varphi_t | ^2] < +\infty$.
\item ${\cal A}^2$  is the set of real-valued non decreasing RCLL predictable
 processes $A$ with $A_0 = 0$ and $\mathbb{E}(A^2_T) < \infty$.

\item ${\mathbb H}^2$  is the set of ${\mathbb  G}$-predictable processes $Z$ such that
 $
 \| Z\|^2:= \mathbb{E}\Big[\int_0^T|Z_t|^2dt\Big]<\infty \,.
 $
\item  ${\mathbb H}^2_{\lambda}:= L^2( \Omega \times [0,T],{\cal P}, \lambda_tdt)$, equipped with the scalar product $\langle U,V \rangle _{\lambda}:= \mathbb{E}\Big[\int_0^TU_t V_t\lambda_tdt\Big]$, for all 
$U,V$ in ${\mathbb H}^2_{\lambda}$. For each $U \in$ ${\mathbb H}^2_{\lambda}$, we set 
$\| U\|_{\lambda}^2:=\mathbb{E}\Big[\int_0^T|U_t|^2\lambda_tdt \Big]<\infty \,.$
\end{itemize}


Since  $\lambda$ vanishes after $\vartheta$, we can suppose that for each 
$U$ in ${\mathbb H}^2_{\lambda}$ $=$ $L^2( \Omega \times [0,T],{\cal P}, \lambda_tdt)$, $U$ (or its representant in ${\cal L}^2( \Omega \times [0,T],{\cal P}, \lambda_tdt)$ still denoted by 
$U$) vanishes  after $\vartheta$. 
\\
Moreover,  $\mathcal{T}$ denotes the set of
stopping times $\tau$ such that $\tau \in [0,T]$ a.s.\, and for each $S$ in $\mathcal{T}$, 
   $\mathcal{T}_{S}$ is  the set of
stopping times
$\tau$ such that $S \leq \tau \leq T$ a.s.

Recall that in this setup, we have a martingale representation theorem with respect to $W$ and $M$ (see e.g. Lemma 1 in \cite{DQS4}).


 We consider  a financial market which consists of one risk-free asset, with  price process $S^{0}$ satisfying $dS_t^{0}=S_t^{0} r_tdt$, and two risky assets with price processes $S^1,S^2$ evolving according to the following equations:
\begin{equation*}
\begin{cases}
dS_t^{1}=S_t^{1}[\mu_t^1dt +  \sigma^1_tdW_t]\\
dS_t^{2}=S_{t^-}^{2} [\mu^2_tdt+\sigma^2_tdW_t-dM_t].
\end{cases}
\end{equation*}
The process $S^0= (S_t^{0})_{0\leq t \leq T}$ corresponds to the price of a non risky asset with interest rate 
process $r= (r_t)_{0\leq t \leq T}$, 
$S^1= (S_t^{1})_{0\leq t \leq T}$ to a non defaultable risky asset, and $S^2= (S_t^{2})_{0\leq t \leq T}$ to a defaultable  asset with total default. The price process $S^2$ vanishes after $\vartheta$.
 
All the processes $\sigma^1,\sigma^2,$ $r, \mu^1,\mu^2$ are 
predictable (that is  ${\cal P}$-measurable). 
We suppose that the coefficients $\sigma^1, \sigma^2 > 0$, and  $r$, $\sigma^1,\sigma^2,$ $\mu^1,\mu^2,\lambda,$ $\lambda^{-1}$,${(\sigma^1)}^{-1}$, 
${(\sigma^2)}^{-1}$ are bounded. 

We consider an investor, endowed with an initial wealth equal to $x $, who can invest his wealth in the three assets of the market. 
At each time $t < \vartheta$, he chooses   the amount $\varphi_t^1$ (resp. $\varphi_t^2$) of wealth invested in the first 
(resp. second) risky asset. However, after time $\vartheta$, he cannot invest his wealth in the defaultable 
asset since its price is equal to $0$, and he only chooses  the amount $\varphi_t^1$ of wealth invested in the first risky asset. Note that the process $\varphi^2$ can be defined on the whole interval $[0,T]$ by setting $\varphi_t^2=0$ for each $t \geq \vartheta$.
A process $\varphi_.= (\varphi_t^1, \varphi_t^2)'_{0 \leq t \leq T}$ is called a {\em risky assets stategy} if 
it belongs to ${\mathbb H}^2 \times  {\mathbb H}^2_{\lambda}$. 
The value of the associated portfolio  (also called {\em wealth}) at 
time $t$ is denoted  by  $V^{x, \varphi}_t$ (or simply $V_t$).

\paragraph{The perfect market model.}
In the classical case of a perfect market model,  the wealth process and the strategy satisfy the self financing condition:
\begin{equation}\label{portfolio}
dV_t  = (r_t V_t+\varphi_t^1 (\mu^1_t - r_t)+\varphi_t^2(\mu^2_t - r_t) ) dt + 
(\varphi_t^1 \sigma^1_t + \varphi_t^2 \sigma^2_t) dW_t - \varphi_t^2  dM_t.
\end{equation}
Setting $K_t:=- \varphi_t^2$, and $Z_t:=\varphi_t^1 \sigma^1_t + \varphi_t^2 \sigma^2_t$, 
we get 
\begin{align*}
dV_t  = (r_t V_t+ Z_t  \theta_t^1+ K_t \theta_t ^2 \lambda_t) dt  
 + Z_t dW_t + K_t  dM_t,
\end{align*}
where $\theta_t^1:=\dfrac{\mu_t^1-r_t}{\sigma_t^1}$ and $\theta_t^2:=  \dfrac{\sigma_t^2 \theta_t^1 -\mu_t^2+r_t}{\lambda_t  }\,{\bf 1}_{\{t \leq \vartheta \} }$. \\
Consider a European contingent claim with maturity $T>0$ and payoff $\xi$ which is $\mathcal{G}_T$ measurable, belonging to  ${L}^2$. The problem is to price and hedge this claim by constructing a replicating portfolio. 
From  \cite[Proposition 2.6 ]{DQS4}, there exists an unique process $(X, Z, K) \in \mathcal{S}^2 \times {\mathbb H}^2 \times  {\mathbb H}^2_{\lambda}$ solution of the following  BSDE with default jump:
\begin{equation}\label{portfolio}
- dX_t = \displaystyle -  (r_t X_t+Z_t \theta_t^1+K_t  \theta_t^2 \lambda_t) dt -  Z_t dW_t - K_t  dM_t\,; \quad
X_T=\xi.
\end{equation}
The solution $(X, Z, K)$ provides the replicating portfolio. More precisely, 
the process $X$  corresponds to its value, and 
the hedging risky assets stategy  $\varphi \in {\mathbb H}^2_{\lambda}$ is given by $\varphi=\Phi (Z, K)$, 
where $\Phi$ is the one to one map defined on ${\mathbb H}^2 \times  {\mathbb H}^2_{\lambda}$ by:
\begin{definition}\label{stbis}
Let  
$\Phi:{\mathbb H}^2 \times  {\mathbb H}^2_{\lambda} \rightarrow {\mathbb H}^2 \times  {\mathbb H}^2_{\lambda}$ be the one-to-one map defined for each $(Z,K) \in {\mathbb H}^2 \times  {\mathbb H}^2_{\lambda} $ by $\Phi (Z, K):= \varphi,$ where $\varphi= (\varphi^1, \varphi^2)$ is given by 
\begin{equation*}
 \varphi_t^{2} = - {K_t} \;\; ; \;\; 
\varphi_t^{1} = 
 \frac{Z_t +   \sigma^2_t  K_t\, }{\sigma^1_t},
\end{equation*}
which is equivalent to 
$
K_t=- \varphi_t^2\, ;\,\,\,
Z_t= {\varphi ^1_t} \sigma^1_t + {\varphi^2_t}\, \sigma^2_t = {\varphi_t}' \sigma_t .
$ 
\end{definition}
Note that the processes $\varphi^2$ and $K$, which belong to ${\mathbb H}^2_{\lambda}$,
both vanish after time $\vartheta$.

The process $X$ coincides with $V^{X_0, \varphi}$, the value of the  portfolio 
associated with initial wealth $x=X_0$ and portfolio strategy $\varphi$. 
From the seller's point of view, this portfolio is a hedging portfolio.  Indeed, by investing the initial amount $X_0$ in the reference assets along the strategy $\varphi$, the seller  can  pay the amount $\xi$ to the buyer at time $T$ (and similarly at each initial time $t$). 
We derive that $X_t$ is the price at time $t$ of the option, called {\em hedging price}, and denoted by 
$X_t(\xi)$. 
By the representation property of  the solution  of a $\lambda$-linear BSDE with default jump (see \cite[Theorem 2.13]{DQS4}), we have that the solution $X$ of BSDE \eqref{portfolio} can be written as follows:  
\begin{equation}\label{free}
X_t(\xi)=\mathbb{E}[e^{-\int_t ^T r_s ds} \zeta_{t,T}\xi \,|\,{\cal G}_t], 
\end{equation} where 
$\zeta_{t, \cdot}$ satisfies $d\zeta_{t,s}= \zeta_{t,s^-} [-\theta^1_s dW_s - \theta^2_s  dM_s]$ with  $\zeta_{t,t}=1.$
 This defines a {\em linear} price system $X$: $\xi \mapsto X (\xi)$.
Suppose now that 
\begin{equation}\label{cth}
\theta^2_t < 1, \; 0 \leq t \leq \vartheta\, \,dt \otimes dP -a.s.\end{equation}
Then $\zeta_{t,\cdot}>0$.
Let $Q$ be the probability measure which admits  $\zeta_{0,T}$ as density on ${\cal G}_T$.
Using Girsanov's theorem, it can be shown that $Q$ is the unique martingale probability measure. 
 In this case, the price system $X$ is increasing 
and corresponds to the classical
 arbitrage free price system (see \cite{ BCJR, JYC}).

\paragraph{The imperfect market model ${\cal M}^g$. } From 
 now on,   we assume that  there are  imperfections in the market which are taken into account via 
the {\em nonlinearity} of the
dynamics of the wealth. More precisely, the
dynamics of the wealth $V$ associated with strategy $\varphi=(\varphi^1, \varphi^2)$  can be written via  a {\em nonlinear} 
driver, defined as follows: 
\begin{definition}[Driver, $\lambda$-{\em admissible} driver]\label{defd}
A  function $g$
is said to be a {\em driver} if\\
$g: [0,T]  \times \Omega \times \R^3  \rightarrow \R $; 
$(\omega, t,y, z, k) \mapsto  g(\omega, t,y,z,k) $
  is $ {\cal P} \otimes {\cal B}(\R^3) 
- $ measurable, and such that
 $g(.,0,0,0) \in {\mathbb H}^2$.
A driver $g$ is called a $\lambda$-{\em admissible driver} if moreover there exists a constant $ C \geq 0$ such that 
$dP \otimes dt$-a.s.\,,
for each $(y_1, z_1, k_1)$, $(y_2, z_2, k_2)$,
\begin{equation}\label{lip}
|g(\omega, t, y_1, z_1, k_1) - g(\omega, t, y_2, z_2, k_2)| \leq
C ( |y_1 - y_2| +|z_1 - z_2| +   \sqrt \lambda_t |k_1 - k_2 |).
\end{equation}
The positive real $C$ is called the $\lambda$-{\em constant} associated with driver $g$.
\end{definition}
Note that condition \eqref{lip} implies that  for each $\,t > \vartheta$, since $\lambda_t=0$,
$g$ does not depend on $k$. 
In other terms, for each $(y,z,k)$, we have: 
$g(t,y,z,k)= g(t,y,z,0)$, $ t > \vartheta$ $dP \otimes dt$-a.s.
%
Let $x \in {\mathbb R}$ be the initial wealth and let $\varphi=(\varphi^1, \varphi^2)$ in ${\mathbb H}^2 \times  {\mathbb H}^2_{\lambda}$ be a portfolio strategy.
We suppose that  the associated {\em wealth} process  $V^{x, \varphi}_t$ (or simply $V_t$)
satisfies  the following dynamics:
 \begin{equation}\label{weaun}
-dV_t= g(t,V_t, {\varphi_t}' \sigma_t , - \varphi_t^{2} ) dt - {\varphi_t}' \sigma_t dW_t +\varphi_t^{2} dM_t, 
 \end{equation}
 with $V_0=x$. Since $g$ is Lipschitz continuous with respect to $y$, this formulation makes sense. Indeed, setting $ f^1_t := \int_0^t {\varphi_t}' \sigma_t  dW_s +\varphi_t^{2} dM_s$, for each $\omega$, the deterministic function $
 (V_t^{Y_0, \varphi }(\omega))$ is defined as the unique solution of the following deterministic differential equation:
 \begin{equation}\label{riun}
V_t^{x, \varphi }(\omega) = x-\int_0^t 
g(\omega, s,V_s^{x, \varphi }(\omega),{\varphi_s}' \sigma_s(\omega), - \varphi_s^{2}(\omega) )ds + f^1_t(\omega)
, \,\, 
0 \leq t \leq T.
 \end{equation}
Note that, equivalently, setting $Z_t= {\varphi_t}' \sigma_t$ and
  $K_t= -  \varphi_t^2 $, the dynamics \eqref{weaun} of the wealth process $V_t$ can be written as follows:
 \begin{equation}\label{wea}
-dV_t= g(t,V_t, Z_t,K_t ) dt -  Z_t dW_t - K_t dM_t.
\end{equation}
In the following, our imperfect market model is denoted by ${\cal M}^g$.\\
Note that in the case of a perfect market (see \eqref{portfolio}), we have:
\begin{equation}\label{perfectlineaire}
g(t,y,z,k) = - r_t y -  \theta^1_t z  -    \theta^2_t  k \lambda_t
 ,
 \end{equation}
 which is a $\lambda$-admissible driver since by the assumptions on the coefficients of the model, the processes $\theta^1$ and $\theta^2$ are bounded.

 %
\subsection{Nonlinear pricing system $\mathcal{E}^g$}
 Pricing and hedging  European options in the imperfect market ${\cal M}^g$ leads to BSDEs with nonlinear driver $g$ and a default jump. By   \cite[Proposition 2.6]{DQS4}, 
we have 
\begin{proposition} \label{existence} Let  $g$ be a $\lambda$-admissible driver, let $\xi \in {L}^2({\cal G_T})$.
There exists an unique solution  $(X(T, \xi), Z(T, \xi), K(T, \xi))$ (denoted simply by
 $(X, Z, K)$)  in $ \mathcal{S}^2 \times {\mathbb H}^2 \times  {\mathbb H}^2_{\lambda}$ of the following BSDE:
\begin{equation}\label{BSDE}
-dX_t = g(t,X_t, Z_t,K_t ) dt  -  Z_t dW_t - K_t dM_t; \quad
X_T=\xi.
\end{equation}
\end{proposition}
Let us consider a European option with maturity $T$ and terminal payoff  $\xi \in {L}^2({\cal G_T})$ in this market model. Let $(X, Z, K)$ be 
the solution of BSDE \eqref{BSDE}.
The process $X$ is equal to  the wealth process associated with initial value $x= X_0$,
strategy $\varphi $ $= \Phi  ( Z,K)$ (where $\Phi$ is defined in Definition \ref{stbis})  that is
 $X= V^{X_0, \varphi}$.
Its initial value $X_0=X_0(T, \xi)$  is thus a sensible price 
 (at time $0$)  of the claim $\xi$ for the seller since this amount allows him/her to construct a trading 
strategy  $\varphi $ $\in {\mathbb H}^2 \times  {\mathbb H}^2_{\lambda}$, called {\em hedging strategy} (for the seller),  such that the value of the associated portfolio is equal to $\xi$ at time $T$. 
Moreover, by the uniqueness of the solution of BSDE \eqref{BSDE}, it is the unique price (at time $0$) which satisfies this hedging property. Similarly, $X_t=X_t(T, \xi)$ satisfies an analogous property at time $t$, and is called the {\em hedging price}
 at time $t$.
 %
This  leads  to a {\em nonlinear pricing} system, first introduced in \cite{EQ96} (also called 
 {\em $g$-evaluation} in 
\cite{Peng2004}) and denoted by ${\cal E}^g$.
For each $S\in [0,T]$, for each $\xi \in {L}^2({\cal G_S})$ 
 the associated 
$g$-evaluation is defined by 
${\cal E}_{t,S}^{^{g}} (\xi):= X_t(S, \xi)$ for each $t \in [0,S]$.

In order to ensure the (strict) monotonicity and the no arbitrage property of the nonlinear pricing system ${\cal E}^g$, we make the following assumption (see 
 \cite[Section 3.3]{DQS4}).
\begin{assumption}\label{Royer} 
Assume that there exists a bounded map \begin{equation*}
 {\bf \gamma}:  [0,T]  \times \Omega\times \R^4   \rightarrow  \R \,; \, (\omega, t, y,z, k_1, k_2) \mapsto 
\gamma_t^{y,z,k_1,k_2}(\omega)
\end{equation*}
 ${\cal P } \otimes {\cal B}(\R^4) $-measurable and satisfying $ dP\otimes dt $-a.s.\,, for each $(y,z, k_1, k_2)$ $\in$ $\R^4$,
\begin{equation} \label{critere}
g( t,y,z, k_1)- g(t,y,z, k_2) \geq  \gamma_t^{y,z, k_1,k_2} (k_1 - k_2 )  \lambda_t,
\end{equation} 
and $P$-a.s.\,, for each $(y,z, k_1, k_2)$ $\in$ $\R^4$,
$\gamma_{t}^{y,z, k_1, k_2} > -1$.

\end{assumption}
\noindent 
This assumption is satisfied e.g. when 
 $g$ is ${\cal C}^1$ in $k$  with $ \partial_k g(t, \cdot) > - \lambda_t$ on $\{t \leq \vartheta\}$. In the special case of a perfect
  market, $g$ is given by \eqref{perfectlineaire}, which implies  that $ \partial_k g(t, \cdot)=- \theta^2_t\lambda_t$. 
 Assumption \ref{Royer}  is then equivalent to $ \theta^2_t<1$ (which corresponds to the usual assumption \eqref{cth}).
 
\begin{remark}\label{prixnul}
Assume that $g(t,0,0,0)=0$ $dP\otimes dt $-a.s.\, 
Then for all  $S\in [0,T]$, ${\cal E}^{^{g}}_{\cdot, S} (0)= 0$ a.s.
Moreover,  by the comparison theorem for BSDEs with default jump  (see \cite[Theorem 2.17]{DQS4}),
the nonlinear pricing system  ${\cal E}^{^{g}}$ is nonnegative, that is, for all $S\in [0,T]$, 
for all $\xi \in {L}^2({\cal G_S})$, if $\xi \geq 0$ a.s., 
then ${\cal E}^{^{g}}_{\cdot, S} (\xi)\geq 0$ a.s.
\end{remark}

 \begin{definition}\label{defmart}
Let $Y \in S^2$. The process $(Y_t)$ is said to be a strong ${\cal E}$-supermartingale (resp. martingale)  if ${\cal E}_{\sigma ,\tau}(Y_{\tau}) \leq Y_{\sigma}$ (resp. $= Y_{\sigma}$) a.s. on $\sigma \leq \tau$,  for all $ \sigma, \tau \in \mathcal{T}_0$. 
\end{definition}

Note that, by the flow property of BSDEs, for each $S\in [0,T]$ and for each $\xi \in {L}^2({\cal G_S})$, the  $g$-evaluation
${\cal E}_{\cdot,S}^g (\xi)$ is an $\mathcal{E}^g$-martingale. Moreover, since $V_t^{x, \varphi}= {\cal E}_{t,T}^g (V_T^{x, \varphi})$, we have:
\begin{proposition}  \label{rima}
For each $x \in \mathbb{R}$ and each portfolio strategy $\varphi$ $\in$ ${\mathbb H}^2\times 
{\mathbb H}^2_{\lambda}$, the associated wealth process $V^{x, \varphi}$ is an $\mathcal{E}^g$-martingale.
\end{proposition}

\begin{example}[Large investor seller] \label{eximp}



Suppose  that the seller of the option  is a  large trader whose hedging strategy $\varphi$ and its associated cost $V$ may influence the market prices
  (see e.g. \cite{CM, BK2}). 
 Taking into account  the possible feedback effects in the market model, the large trader-seller  may suppose that the coefficients 
are of the form $\sigma_t (\omega)= \bar \sigma (\omega, t, V_t, \varphi_t)$ where $
\bar \sigma:   \Omega \times [0,T]  \times \mathbb{R}^3 \mapsto  \mathbb{R}^2  \,; \, ( \omega,t, x, z,k) \mapsto 
 \bar \sigma ( \omega, t, x, z, k)
$ is a ${\cal P } \otimes {\cal B}({\mathbb R}^3)/{\cal B}({\mathbb R}^2)$-measurable map,
 and similarly for the other coefficients $ r$, $ \mu^1$, $ \mu^2$. The driver is thus of the form:
 $$g(t, V_t, \varphi '_t \bar \sigma_t (t,  V_t, \varphi_t), -\varphi_t^{2}  )=- \bar r (t,  V_t, \varphi_t) V_t-\varphi_t^1 \,
 (\bar \mu^1_t - \bar r_t) (t,  V_t, \varphi_t)-\varphi_t^2(\bar \mu^2_t - \bar r_t)   (t,  V_t, \varphi_t). $$ 
%
  Here, the map  $\Psi:$ $(\omega, t,y,\varphi) \mapsto (z,k)$ with $z={\varphi}' \bar \sigma_t(\omega,t,y,\varphi)$ and
 $k=- \varphi^2$ is assumed to be one to one with respect to $\varphi$, and such that its inverse $\Psi^{-1}_{\varphi} $ is ${\cal P}\otimes {\cal B} ({\mathbb R}^3)/{\cal B} ({\mathbb R}^2)$-measurable. 

 We can also include the case when the seller strategy influences the default probability via the default intensity by considering the driver  of the following form 
\footnote{ For details, see in \cite{DQS4} the example of the large investor seller, in particular equation  (3.12).}
 $$g(t, V_t, \varphi '_t \bar \sigma_t (t,  V_t, \varphi_t), -\varphi_t^{2}  )-\gamma(t,V_t,\varphi_t)\lambda_t \varphi_t^2,$$ 
 where $\gamma:\Omega \times [0,T] \times \mathbb{R}^3 \mapsto \mathbb{R}^2$ is ${\cal P}\otimes {\cal B} ({\mathbb R}^3)/{\cal B} ({\mathbb R}^2)$-measurable.


\end{example}
For other examples (the case of taxes and the case of different borrowing and lending interest rates), the reader is referred to \cite{DQS5} or \cite{EPQ01}.

 Note that when the market is perfect,  the prices 
 $S^0, S^1$ and $S^2$ are ${\cal E}^g$-martingales.\footnote{This  corresponds to the well-kown property
  that the discounted prices of the reference assets are martingales under the martingale probability measure $Q$.}
 This is also true in the examples of the large investor and  different borrowing and lending interest rates.
 In the case of taxes, this property is not necessarily satisfied. 

\section{American option pricing from the seller's point of view}\label{SECAM}

Let us consider an American option associated with horizon $T>0$ and a payoff given by an RCLL adapted process $(\xi_t, 0 \leq t \leq T)$. At time $0$, it consists  in the selection 
of a stopping time $\nu \in {\cal T}$ and   the payment of the payoff 
$\xi_{\nu}$ from the seller to the buyer.

The {\em seller's superhedging price} of the American option at time $0$, denoted by $u_0$, is classically defined as the minimal initial capital which enables the seller  to invest in a portfolio which  covers  his liability to pay to the buyer up to $T$ no matter 
what the exercise time chosen by the buyer.
More precisely, for each initial wealth $x$, we denote by ${\cal A} (x)$ the set of all portfolio strategies $ \varphi$ $\in$  ${\mathbb H}^2$ such that 
$V^{x, \varphi}_{ t} \geq \xi_t$, $0 \leq t \leq T$ a.s. 
The {\em seller's superhedging price}
of the American option is thus defined by
$$ u_0:= \inf \{x \in \R,\,\, \exists  \varphi \in {\cal A} (x) \}.$$
  \begin{remark}
Suppose that $g(t,0,0,0) =0$. From Remark \ref{prixnul}, we derive that 
 if $\xi_\cdot \geq 0$, the infimum in the definition of $u_0$ can be taken only over nonnegative initial wealths, that is, $ u_0:= \inf \{x \geq 0,\,\, \exists  \varphi \in {\cal A} (x) \}.$
\end{remark}
 We define the $g$-{\em value} of the American option as
\begin{equation} \label{prixS}
\sup_{\nu \in {\mathcal T}}{\cal E}_{0,\nu}^g ( \xi_{\nu}).
\end{equation}
%

\begin{proposition}[Characterization of the {\em $g$-value}] \label{cara}
 There exists a unique process $(Y,Z,K,A)$ solution of the reflected BSDE 
  associated with driver $g$ and obstacle $\xi$ in 
 ${\cal S}^2 \times {\mathbb H}^2 \times {\mathbb H}^2_{\lambda}\times {\cal A}^2 $,
that is
 \begin{align}
   &-dY_t = g(t,Y_t,  Z_t, K_t )dt +dA_t - Z_t  dW_t -  K_t dM_t; \; \; Y_T = \xi_T, \;    \label{RBSDE} 
   \text{with} &  \\
 &  \,\,\,\,Y \geq \xi \,,\,\,  \displaystyle   \int_0^T (Y_t - \xi_t) dA^c_t = 0 \text{ a.s. and}\,\,
  \Delta A_{t}^{d}= \Delta A_{t}^{d} 
 {\bf 1}_{\{Y_{t^-} = \xi_{t^-}\}},
 \label{sko}
\end{align}
where  $A^c$ denotes the continuous part of $A$ and $A^d$ its discontinuous part.

Moreover, we have 
\begin{equation}\label{dualrepresentation}
 Y_0=\sup_{\nu \in {\cal T} }{\cal E}_{0,\nu} ( \xi_{\nu}).
\end{equation}
\end{proposition}

Note that the equality \eqref{dualrepresentation} was shown in \cite{EQ96} in a Brownian framework with a continuous obstacle, and generalized in \cite{QS2} to the RCLL case with jumps. 

\begin{proof}
Let us first show that there exists a unique solution of RBSDE \eqref{RBSDE}. 
As usual, we first consider the case when the driver $g(t)$ does not depend on the solution. 
By using the representation property of ${\mathbb G}$-martingales 

(see e.g. Lemma 1 in \cite{DQS4}) 

and some results of optimal stopping  theory, one can show, proceeding as in \cite{EPardQ} (see also \cite{Ham} and \cite{QS2}), that there exists a unique solution of the associated 
RBSDE  \eqref{sko}. 
The proof in the general case is the same as for non reflected BSDEs with default jump (see the proof of Proposition 2.6 in \cite{DQS4}). It is based on a fixed point argument and the a priori estimates for RBSDEs with default given in the Appendix (see Lemma \ref{est}).

Let us now show the equality \eqref{dualrepresentation}.
Proceeding as in the proof of  \cite[Theorem 3.3]{QS2} which was given in the framework of a random Poisson measure, we can prove that for each $S \in {\cal T}$,
\begin{equation}\label{dualS}
Y_S= {\rm ess} \sup_{\tau \in {\cal T}_S} \, \cal{E}^g_{S,\tau }[\xi_\tau] \quad {\rm a.s.}
\end{equation}
Equality \eqref{dualrepresentation} then follows by taking $S=0$. 
\end{proof}

\begin{lemma}\label{Ac}
%
If $\xi$ is left-u.s.c. along stopping time, then
$A$ is continuous. 
\end{lemma}

\begin{proof}
Let $\tau$ be a 
predictable stopping time. By \eqref{RBSDE}, we have
$\Delta A_{\tau}=(\Delta Y_{\tau})^-$. Using the Skorokhod conditions \eqref{sko}, we get 
$
\Delta A_{\tau}= {\bf 1}_{ \{Y_{\tau^-} = \xi_{\tau^-}\} }( Y_{\tau}- Y_{\tau^-})^-= 
{\bf 1}_{ \{Y_{\tau^-} = \xi_{\tau^-}\} }( Y_{\tau}- \xi_{\tau^-})^-
$ a.s.
Now, since by assumption, $\xi_{\tau^-} \leq \xi_{\tau}$ a.s.\,, we have 
$Y_{\tau}- \xi_{\tau^-} \geq Y_{\tau}- \xi_{\tau} \geq 0$ a.s. We derive that $\Delta A_{\tau}= 0$ a.s.
It follows that 
$A$ is continuous. 
\end{proof}

We now provide two characterizations of the seller's superhedging price, which generalize those provided in the literature in the  case of a perfect market (see \cite{EPardQ}) to the case of an imperfect market.



\begin{proposition}[{\em Seller's superhedging price} of the American option]\label{americano}
The {\em seller's superhedging price} $u_0$ of the American option is equal to its {\em $g$-value}, that is 
\begin{equation}\label{optimalstopping}
u_0=
\sup_{\nu \in {\mathcal T}}{\cal E}_{0,\nu}^g ( \xi_{\nu}).
\end{equation}
 Moreover, $u_0=Y_0$, where $(Y,Z,K,A)$
 is the solution of the {\em nonlinear} reflected BSDE \eqref{RBSDE} and
the portfolio strategy $\varphi^*:= \Phi(Z,K)$ (where $\Phi$ is defined in Definition \ref{stbis}) is a superhedging strategy for the seller.
\end{proposition}
Note that in the case of a perfect market, 
equality \eqref{optimalstopping} reduces to the well-known characterization of the price of the American option as the value function of a classical optimal stopping problem, and the equality $u_0=Y_0$ corresponds to the well-known
characterization of this price as the solution of the {\em linear} reflected BSDE associated with the {\em linear} driver \eqref{perfectlineaire} (see \cite{EPardQ}).

\begin{proof}
The proof is based on the characterization of the {\em $g$-value}  as  the solution  of the reflected  
BSDE \eqref{RBSDE}
(see Proposition \ref{cara}).
It is sufficient to show that $u_0 =Y_0$ and $\varphi^* \in \mathcal{A}(u_0)$.\\
Let $\mathcal{H}$  be the set of initial capitals which allow the seller to be ``super-hedged", that is
$\mathcal{H}= \{ x \in \mathbb{R}: \exists \varphi \in \mathcal{A}(x) \}. $
Note that $u_0= \inf \mathcal{H}$.

Let us first show that 
\begin{equation} \label{phiA}
\varphi ^*\in \mathcal{A}
(Y_0).
\end{equation}  
By \eqref{weaun}-\eqref{wea}, for each $\omega$, the trajectory of the value of this portfolio $t \mapsto V_t^{Y_0, \varphi ^*}(\omega)$ satisfies the following forward differential equation:
\begin{align}\label{ri}
V_t^{Y_0, \varphi ^*}(\omega) = Y_0-\int_0^t 
g(s,\omega, V_s^{Y_0, \varphi ^{*}}(\omega),Z_s(\omega),K_s(\omega))ds + f^1_t(\omega)
, \,\, 
0 \leq t \leq T,
 \end{align}
 where $ f^1_t := \int_0^t Z_s dW_s +\int_0^t K_s dM_s$. 
Moreover, since $Y$ is the solution of the reflected BSDE \eqref{RBSDE}, for almost every $\omega$, the function  $t \mapsto Y_t(\omega)$ satisfies:
\begin{equation}\label{forww}
Y_t(\omega)=Y_0-\int_0^t g(s,\omega,Y_s(\omega),Z_s(\omega),K_s(\omega))ds+f^2_t(\omega), \,\, 0 \leq t \leq 
T,
\end{equation}
where $ f^2_t := f^1_t -A_t$. 
Let us apply for fixed $\omega$ a comparison result for forward differential equations 

(see \cite{DQS5} in the Appendix).

We thus get  $V_t^{Y_0, \varphi^*} \geq Y_t $, $0 \leq t \leq T $ a.s. Since 
$Y_. \geq \xi_.$, we have  $V_t^{Y_0, \varphi^*} \geq 
\xi_t $, $0 \leq t \leq T $ a.s.\,, which implies the desired property \eqref{phiA}. It follows that $Y_0 \geq u_0$.

Let us show the converse inequality.
%
%
%
Let $x \in \mathcal{H}$. There exists $\varphi \in \mathcal{A}(x)$ such that 
$V^{x, \varphi}_{t} \geq \xi_t$, $0 \leq t \leq T$ a.s. 
For each $\nu \in \mathcal{T}$ we thus have
$V^{x, \varphi}_{\nu } \geq \xi_\nu $ a.s.\,
By taking the $\mathcal{E}^g$-evaluation in this inequality, using the monotonicity 
of $\mathcal{E}^g$ and the $\mathcal{E}^g$-martingale property of the wealth process $V^{x, \varphi}$,
we obtain
$
x =\mathcal{E}^g _{0,\nu }[V^{x, \varphi}_{\nu}] \geq 
\mathcal{E}_{0, \nu }^g [\xi_\nu].
$
By arbitrariness of $\nu \in \mathcal{T}$, we get
$
x \geq 
\sup_{\nu \in \mathcal{T}} \mathcal{E}_{0, \nu }^g [\xi_\nu],
$
which holds for any $x \in \mathcal{H}$. By taking the infimum over $x \in \mathcal{H}$, 
we obtain $u_0 \geq Y_0$. We derive that $u_0 = Y_0$. By \eqref{phiA}, we thus have $\varphi^* \in \mathcal{A}(u_0)$, which ends the proof. 
\end{proof}

\begin{remark}
In general, except when $g$ does not depend on $y$, by \eqref{forww}, we have
$$ Y_\cdot= Y_0-\int_0^\cdot g(s,Y_s,Z_s,K_s)ds+\int_0^\cdot  Z_sdW_s+\int_0^.K_sdM_s -A_\cdot \not \equiv V_\cdot^{Y_0, \varphi^*} -A_\cdot.$$
\end{remark}

 \begin{remark}\label{Important} We can define the seller's superhedging price of the American option at
  each time/stopping time $S \in {\cal T}$.
   More precisely, for each initial wealth $X\in L^2( {\cal F}_S) $, a {\em super-hedge} against the American option
    is a  portfolio strategy $ \varphi$ $\in$  ${\mathbb H}^2 \times  {\mathbb H}^2_{\lambda}$ such that
$
V^{S,X, \varphi}_{t } \geq \xi_t,$  $S\leq t \leq T$ a.s.\,,
where $V^{S,X, \varphi}$ denotes the wealth process associated with initial time $S$ and initial condition $X$.
The {\em seller's superhedging price} at time $S$ is defined by 
$
 u(S):= {\rm ess} \inf \{X\in L^2( {\cal F}_S),\,\, \exists   \varphi \in {\cal A}_S (X) \},
$ where ${\cal A} _S(X) $ is the  set of all super-hedges associated with initial time $S$ and initial wealth $X$. 
Using equality \eqref{dualS} and similar arguments to those used in 
the proof of Theorem \ref{americano} above, we obtain:
$$u(S)= {\rm ess} \sup_{\nu \in {\cal T}_S} \, \cal{E}^g_{S,\nu }(\xi_\nu) =Y_S \quad {\rm a.s.}$$
where $(Y, Z, K, A)$ is the  solution of 
RBSDE \eqref{RBSDE}. 
\end{remark}
 \begin{remark}\label{arbsel}
 The seller's superhedging price $u_0$ is clearly an {\em upper bound} of the possible prices for the American option. 
Indeed, no  rational agent would pay more than $u_0$
 since there is a cheaper way to achieve at least the same payoff, whatever the exercise time is. 
 Indeed,  by investing the amount $u_0+ \varepsilon$ and  following the  strategy $\varphi^* $, whatever
  the exercise time $\nu$ is, 
 he will make the gain $V^{u_0+ \varepsilon, \varphi^*}_{\nu}  > V^{u_0, \varphi^*}_{\nu} (\geq \xi_{\nu})$ a.s. by a  strict comparison property for deterministic differential equations 

(see \cite{DQS5} in the Appendix).

 Moreover, if the price of the option is equal to $u_0+\varepsilon$, then, by investing this amount following the strategy $\varphi^* $, whatever
  the exercise time $\nu$ is, the seller will make a gain $V_\nu^{u_0+\varepsilon, \varphi^*}-\xi_\nu>0$ a.s. Hence if the price is strictly greater than $u_0$, then there exists an arbitrage for the seller.
 \end{remark}
 
  \begin{remark}
In \cite{KaKou}, it is proved that the {\em seller's superhedging price} of the American option is equal to its {\em $g$-value}
  in  the case of  a higher interest rate for borrowing, 
 by using a dual approach. This approach relies on the convexity properties of the driver and cannot be adapted to our 
 case, except when $g$ is  convex with respect to $(y,z,k)$.
 \end{remark}

Suppose now that the buyer has bought the American option at the selling price $u_0=Y_0$. We address the problem of the choice of her/his exercise time.
 We introduce the following definition.
  \begin{definition} \label{rational exercise}A stopping time $\nu$ $\in {\mathcal T}$ is called a {\em rational exercise time} for the buyer of the American option if it is optimal  for Problem \eqref{optimalstopping}, that is satisfies
 $\sup_{\tau \in {\mathcal T}}{\cal E}_{0,\tau}^g ( \xi_{\tau})= {\cal E}_{0,\nu}^g ( \xi_{\nu})$. \end{definition}

 By the optimality criterium provided in \cite{QS2} (see Proposition 3.5), we have:
\begin{proposition} (Characterization of {\em rational exercise times})\label{II}
Let $\nu$ $\in$ ${\cal T}$. Then, $\nu$ is a a {\em rational exercise time} for the buyer  if and only if $Y_{\nu} = \xi_\nu$ a.s. and $A_{\nu} =0$ a.s.\,, where $(Y,Z,K,A)$
 is the solution of the reflected BSDE \eqref{RBSDE}.
 \end{proposition}




  Suppose now that  the price of the American option is equal to the selling price not only at time $0$ but even at each time 
$S \in {\cal T}$, that is, the price at  time $S$ is equal to $u(S) =Y_S$ (see Remark \ref{Important}). From a financial point of view, this makes sense, in particular when the seller is a large investor (see Example \ref{eximp}).
Suppose that the buyer has bought the American option at time $0$ (at the price $u_0=Y_0$). 
Let us show that it is not profitable for the him/her to exercise his/her option at a stopping time $\nu$ which is not  a rational exercise time. First, it is not in his interest to exercise at a time $t$  such that $Y_t > \xi_t$ since  
he would loose a financial asset (the option) 
with value $Y_t$ while he would receive the lower amount $\xi_t$ by exercising the option. 


Second, it is not in the interest of the option holder to exercise at a stopping time $\nu$ greater than 
$\bar \nu$, defined by 
$\bar \nu := \inf \{ s \geq 0 ,\, A_s \neq 0 \}$. Let us show that 
$Y_{\bar \nu}= V^{Y_0, \varphi^*}_{\bar \nu}$ a.s.
Note that by definition of $\bar \nu$, $A_{\bar \nu}=0$ a.s. 
Hence, 
for a.e. $\omega$, the trajectories $t \mapsto Y_t(\omega)$ and  $t \mapsto V_t^{Y_0, \varphi^*}(\omega)$ are solutions on $[0, 
\bar \nu(\omega)]$ of the same differential equation (with initial value $Y_0$), which implies that they are equal, 
by uniqueness of the solution. Hence, $Y_{\bar \nu}= V^{Y_0, \varphi^*}_{\bar \nu}$ a.s.
Without loss of generality, we can suppose that  for each $\omega$, we have $Y_{\bar \nu}  (\omega) = V^{Y_0, \varphi^*}_{\bar \nu} (\omega) $. 
Let $\nu \geq  \bar \nu$. Let $B:=  \{\nu  > \bar \nu \}$. Suppose that $ P(B)>0$. 
Hence, 
$A_{\nu}>0$ a.s. on $B$.
Then, by selling the option at time $\bar \nu $, the option holder receives the amount 
$Y_{\bar \nu}$  which he can invest in the market along the strategy  $\varphi^*$. Since
$Y_{\bar \nu}= V^{Y_0, \varphi^*}_{\bar \nu}$, by the flow property of the forward differential equation \eqref{weaun} with $\varphi = \varphi^*$ and $x= Y_0$, the value of the associated portfolio is equal at time $\nu$ to 
$V^{Y_{\bar \nu}, \varphi^*}_{\nu}= V^{V^{Y_0, \varphi^*}_{\bar \nu}, \varphi^*}_{\nu} =V^{Y_0, \varphi^*}_{\nu}$.
Since $A_{\nu}>0$ on $B$, by the strict comparison result for forward differential equations applied to \eqref{ri} and \eqref{forww} 

(see \cite{DQS5} in the Appendix),

 we get
 $V^{Y_0, \varphi^*}_{ \nu} > Y_{ \nu}$ a.s. on $B$, which implies $V^{Y_{\bar \nu}, \varphi^*}_{\nu}=V^{Y_0, \varphi^*}_{ \nu} > \xi_{ \nu}$ a.s. on $B$. We thus have 
 $V^{Y_{\bar \nu}, \varphi^*}_{\nu} \geq \xi_{ \nu}$ a.s. with $P(V^{Y_{\bar \nu}, \varphi^*}_{\nu} > \xi_{ \nu}) >0$. Hence, at time $\bar \nu$, it is more interesting for the buyer to sell immediately his option than to keep it and to exercise it after.

\begin{proposition}\label{llI}(Existence of {\em rational exercise times})
Suppose  that the payoff $\xi$ is left u.s.c. along 
stopping times.
 Let $\nu^* := \inf \{ s \geq 0 ,\, Y_s = \xi_s \}$ and $\bar \nu := \inf \{ s \geq 0 ,\, A_s \neq 0 \}$. \\
The stopping time $\nu^*$ (resp. $\bar \nu$) is the  minimal (resp. maximal) rational exercise time.
\end{proposition}


\begin{proof}
The right continuity of $(Y_t)$ and $(\xi_t)$ ensures that $Y_{\nu^*} = \xi_{\nu^*}$ a.s. 
By definition of $\nu^*$, we have $Y_t >\xi_t$ a.s. on $[0, \nu^*[$. Hence the process $A$ is constant on $[0, \nu^*[$ 
and even on $[S, \nu^*]$ because $A$ is continuous (see Lemma \ref{Ac}).

By Proposition \ref{II}, $\nu^*$ is thus a rational exercise time  and is the  minimal one.
From the definition of $\bar \nu$,  and the continuity of $A$, we have 
$ A_{\bar \nu}  =0 $ a.s. 
Also, we have a.s. for all 
$ t > \bar \nu,  \; A_t >  A_{\bar \nu} = 0.$
Since $A$ increases only on the set $\{ Y_\cdot = \xi_\cdot\}$, it follows that 
$Y_{\bar \nu} = \xi_{\bar \nu}.$  By Proposition \ref{II}, 
 $\bar{\nu}$ is a rational exercise time and is the  maximal one.
 \end{proof}

 
When $\xi$ is only RCLL, there does not exist necessarily an {\em rational} exercise time for the buyer. However, we have the following result.
%

\begin{proposition}\label{varep} (Existence of $\varepsilon$-{\em rational} exercise time) 
Suppose 
$\xi$ is RCLL. 
For each $\varepsilon>0$, the stopping time 
$\nu_{\varepsilon}:= \inf \{t \geq 0: \,\,\, Y_t \leq \xi_t+\varepsilon \}$ satisfies
 \begin{equation}\label{fifi}
 \sup_{\tau \in {\mathcal T}}{\cal E}_{0,\tau}^g ( \xi_{\tau}) \,\,\leq \,\,  {\cal E}_{0,\nu_{\varepsilon}}^g ( \xi_{\nu_{\varepsilon}})  + K 
 \varepsilon \quad   \rm{a.s.}\,,
 \end{equation}
where $K$ is a constant which only depends on $T$ and the Lipschitz constant  $C$ of $g$.
In other words, $\nu_{\varepsilon}:= \inf \{t \geq 0: \,\,\, Y_t \leq \xi_t+\varepsilon \}$ is a 
$K\varepsilon$-{\em rational} exercise time.

\end{proposition}
For the proof, the reader is referred to Theorem 3.2 in \cite{QS2}. 

\section{The buyer's point of view}\label{SECB}
Let us consider the pricing and hedging problem of a European option with maturity $T$ and payoff $\xi \in L^2({\cal G}_T)$ from 
the buyer's point of view. Supposing the initial price of the option is $z$, he starts with the amount $-z$ at time $t=0$, and looks to find a risky-assets strategy $\tilde \varphi$ such that the payoff  that he receives at time $T$ allows him to recover the debt he incurred at time $t=0$ by buying the option, that is such that 
$V^{-z, \tilde \varphi}_T + \xi=0 \quad {\rm a.s.}\,$
or equivalently, $V^{-z, \tilde \varphi}_T =- \xi$ a.s.\,

The buyer's superhedging price of the option is thus equal to 
the opposite of the seller's superhedging price of the  option with payoff  $-\xi$, that is $-{\cal E}_{0,T}^{^{g}} (-\xi)=- \tilde X_0$, where
$( \tilde X,  \tilde Z, \tilde K)$ is the solution of the BSDE associated with driver $g$ and terminal condition $-\xi$. Let us specify the hedging strategy for the buyer. Suppose that the initial price of the option is $z:= - \tilde X_0$. The process $\tilde X$ is equal to  the value of the portfolio associated with initial value $-z= \tilde X_0$ and
strategy $\tilde \varphi $ $:= \Phi  
( \tilde Z,\tilde K)$ (where $\Phi$ is defined in Definition \ref{stbis})  that is
 $\tilde X= V^{\tilde X_0, \tilde \varphi}= V^{-z, \tilde \varphi}$. Hence, $V^{-z, \tilde \varphi}_T = \tilde X_T= -\xi$ a.s.\,, which yields that  $\tilde \varphi$ 
is the hedging risky-assets strategy for the buyer. Similarly, $-{\cal E}_{t,T}^{^{g}} (-\xi)=- \tilde X_t$ satisfies an analogous property at time $t$, and is is called the {\em hedging price for the buyer}
 at time $t$.


This leads to the {\em nonlinear pricing system $\tilde {\cal E}^{^{g}}$ relative to the buyer} in the market ${\cal M}^g$
 defined for each $(S, \xi) \in [0,T]\times L^2({\cal G}_S)$ by 
 \begin{equation}\label{tildeE}
 \tilde {\cal E}^{^{g}}_{\cdot, S}(\xi):=
-{\cal E}^{^{g}}_{\cdot, S} (-\xi).
\end{equation}
%
%

%
\begin{remark} \label{perfectegal}
%

Note that $\tilde {\cal E}^{^{g}}_{\cdot, S}(\xi)$ is equal to the solution of the BSDE with driver $-g(t,-y,-z,-k)$ and terminal condition $\xi$. Hence, if we suppose that 
$-g(t,-y,-z,-k) \leq g(t,y,z,k)$ (which is satisfied if, for example, $g$ is convex with respect to $(y,z,k)$), then, by the comparison theorem for BSDEs, we have  $\tilde {\cal E}^{^{g}}_{\cdot, S}(\xi)= -{\cal E}^{^{g}}_{\cdot, S} (-\xi) \leq  {\cal E}^{^{g}}_{\cdot, S}(\xi)$  for each $(S, \xi) \in [0,T]\times L^2({\cal G}_S)$.

Moreover, when $-g(t,-y,-z,-k) = g(t,y,z,k)$ (which is satisfied if, for example, $g$ is linear with respect to $(y,z,k)$, as in the perfect market case), we have $\tilde {\cal E}^{^{g}}=  {\cal E}^{^{g}}$. 
%
\end{remark}

We now introduce the definition of the buyer's superhedging price of the American option.
\begin{definition}
The buyer's superhedging price of the American option with payoff $\xi$ is defined as \footnote{
We have $(0,0) \in {\cal B}(\xi_0)$. Hence, $\tilde u_0\geq \xi_0$.
Moreover, if $g(t,0,0,0)=0$ and $\xi_0\geq 0$, then 
$\tilde u_0= \sup\{x \geq 0, \,\,\, \exists (\tau, \varphi) \in\mathcal{B}(x)\}$.} 
$$v_0=\sup\{x \in \mathbb R, \,\,\, \exists (\tau, \phi) \in \mathcal{B}(x)\},$$
where $\mathcal{B}(x):= \{(\tau, \varphi) \in \mathcal{T} \times \mathbb{H}_2 \times \mathbb{H}_2^\nu \text{ such that } V_\tau^{-x, \varphi}+\xi_\tau \geq 0 \text{ a.s. }\}.$
\end{definition}
%

Let us now provide a characterization of this price which requires an additional regularity assumption on the payoff.

\begin{proposition} \label{americano}
Suppose that $(\xi_t)$ is left upper semicontinuous along stopping times. 
The buyer's superhedging price $v_0$ of the American option satisfies: 
\begin{equation}\label{optimalstopping}
v_0=
-\inf_{\nu \in {\mathcal T}}{\cal E}_{0,\nu}^g ( -\xi_{\nu}).
\end{equation}
 Moreover, $v_0=-\Tilde{Y}_0$, where $(\Tilde{Y},\Tilde{Z},\Tilde{K}, \Tilde{A})$
 is the solution of the  reflected BSDE associated with driver $g$ and upper obstacle $-\xi$, that is,
\begin{align}
   &-d\Tilde{Y}_t = g(t,\Tilde{Y}_t,  \Tilde{Z}_t, \Tilde{K}_t )dt -d\Tilde A_t - \Tilde{Z}_t  dW_t -  \Tilde{K}_t dM_t; \; \; \Tilde{Y}_T = -\xi_T, \;    \label{T} 
   \text{with} &  \\
 &  \,\,\,\,\Tilde{Y} \leq -\xi \,,\,\,  \displaystyle   \int_0^T (\Tilde{Y}_t + \xi_t) d\Tilde A_t = 0 \text{ a.s. }\,,
  \nonumber
\end{align}
where the non decreasing process $\Tilde A$ is continuous.
\footnote{This property follows from the left regularity assumption made on $\xi$. It can be shown by using similar arguments as those used in Remark \ref{Ac}.
}.
Let $\Tilde{\tau}:=\inf \{t \geq 0:\,\,\, \Tilde{Y}_t=-\xi_t   \}$ and $\Tilde{\varphi}:= \Phi(\Tilde{Z},\Tilde{K})$ (where $\Phi$ is defined in Definition \ref{stbis}). 
We have $(\Tilde{\tau}, \Tilde{\varphi})$ $\in {\cal B}(v_0)$.

\end{proposition}

\begin{proof}

Let us first note that $\inf_{\nu }{\cal E}^{g}_{0,\nu} ( -\xi_{\nu})$ is characterized as the solution of the reflected BSDE \eqref{RBSDE} (by Proposition \ref{cara}).
We thus have to show that $v_0 =-\Tilde{Y}_0$.\\
Set
$\mathcal{S}:= \{ x \in \mathbb{R}: \exists (\tau,\varphi) \in \mathcal{B}(x) \}$. 

Let us first show that $-\Tilde{Y}_0 \leq v_0$. Since $v_0= \sup \mathcal{S}$, 
it is sufficient to show that $-\Tilde{Y}_0 \in \mathcal{S}$. To this aim, we prove that 
\begin{equation}\label{bbb}
(\Tilde \tau,\Tilde \varphi) \in \mathcal{B}
(-\Tilde{Y}_0).
\end{equation}
We consider the portfolio associated with the initial capital ${\Tilde{Y}_0}$ and the strategy 
$\Tilde{\varphi}$. Its value $V^{\Tilde{Y}_0,\Tilde{\varphi}}$ satisfies:
$$
V_t^{\Tilde{Y}_0, \Tilde{\varphi}} = \Tilde{Y}_0-\int_0^t 
g(s,V_s^{\Tilde{Y}_0, \Tilde{\varphi}},\Tilde{Z}_s,\Tilde{K}_s)ds + \int_0^t \Tilde{Z}_s dW_s +\int_0^t \Tilde{K}_s dM_s, \,\, 
0 \leq t \leq T.
 $$


Now, by definition of $\Tilde{\tau}$, we have that almost surely on $[0,\Tilde{\tau}[$, $\Tilde{Y}_t<-\xi_t$, which implies that the process $\Tilde A$ is constant on $[0,\Tilde{\tau}[$. Since $\Tilde A$ is continuous, we derive that $\Tilde A$ is equal to $0$ on $[0,\Tilde{\tau}]$.
We thus get the equality $V_{\Tilde{\tau}}^{\Tilde{Y}_0, \Tilde{\varphi}}= \Tilde{Y}_{\Tilde{\tau}}$ a.s. Moreover, the right-continuity of the processes $(\Tilde{Y}_t)$ and $(\xi_t)$ ensures that $\Tilde{Y}_{\Tilde{\tau}}= -\xi_{\Tilde{\tau}} \text{ a.s. }$
We thus conclude that
$V_{\Tilde{\tau}}^{\Tilde{Y}_0, \Tilde{\varphi}}=-\xi_{\Tilde{\tau}}$ a.s.\,
The desired property \eqref{bbb} follows.
Hence, we have $-\Tilde{Y}_0$ $\in$ $\mathcal{S}$, and thus $-\Tilde{Y}_0 \leq v_0$.

 It remains to prove that $v_0 \leq -\Tilde{ Y}_0$.
Let $x \in \mathcal{S}$. By definition, there exists $(\tau,\varphi) \in \mathcal{B}(x)$ such that 
$V^{-x, \varphi}_{t} \geq -\xi_\tau $ a.s. 
By taking the $\mathcal{E}^g$-evaluation, using the monotonicity 
of $\mathcal{E}^g$ and the $\mathcal{E}^g$-martingale property of the wealth process $V^{-x, \varphi}$,
we derive that
$
-x =\mathcal{E}^g _{0,\tau }(V^{-x, \varphi}_{\tau}) \geq 
\mathcal{E}_{0, \tau }^g (-\xi_\tau).
$
We thus obtain the inequality
$
x \leq 
-\inf_{\tau \in \mathcal{T}} \mathcal{E}_{0, \tau }^g (-\xi_\tau)=-\Tilde{Y}_0,
$
which holds for any $x \in \mathcal{S}$. By taking the supremum over $x \in \mathcal{S}$, 
we get $v_0 \leq -\Tilde{Y}_0$. It follows  that $v_0 =-\Tilde{ Y}_0$. By \eqref{bbb}, we get $(\Tilde{\tau},\Tilde{\varphi}) \in \mathcal{B}(v_0)$, which completes the proof.
\end{proof}

 Note that if the price of the option is equal to $v_0-\varepsilon$, then, by investing the amount $-v_0+\varepsilon$ following the strategy $\Tilde{\varphi}$, then, whatever
  the exercise time $\nu$ chosen, the buyer will make a gain $V_\nu^{-v_0+\varepsilon, \Tilde{\varphi}}+\xi_\nu>V_\nu^{-v_0, \Tilde{\varphi}}+\xi_\nu \geq 0$ a.s. Hence if the price is strictly smaller than $v_0$, then there exists an arbitrage for the buyer.
 
 Suppose now that 
$-g(t,-y,-z,-k) \leq g(t,y,z,k)$. By Remark \ref{perfectegal}, we get $v_0 \leq u_0$. By Remark \ref{arbsel}, we derive that the interval $[v_0,u_0]$ can be seen as a non-arbitrage interval for the price of the American option in the sense of \cite{KaKou}. In the example of a higher interest rate for borrowing, this result corresponds to the one shown in \cite{KaKou} by a dual approach.
\footnote{Note that in the particular case of a perfect market, the buyer's superhedging price is equal to the seller's superhedging price, that is $v_0=u_0$, since, in this case, the $g$-evaluation $\mathcal{E}^g$ is linear.}

\section{Appendix}
We give here some a priori estimates for RBSDEs with default jump.
\begin{lemma}[A priori estimate for RBSDEs] \label{est}
Let $f^1$ be a $\lambda$-{\em admissible} driver with $\lambda$-constant $C$ and let  $f^2$ be a driver. 
Let $\xi$   be an adapted RCLL processes.
 For $i=1,2$, let $(Y^i, Z^i ,K^i, A^i)$  be  a solution of the RBSDE associated with  
terminal time $T$, driver $f^i$  and obstacle $\xi$. 
 Let $ \eta, \beta >0 $ be such that 
 $\beta \geq \frac{3}{\eta} +2C $ 
and $\eta \leq \frac{1}{C^2}$. \\
Let $\bar f(s): = f^1(s, Y^2_s, Z^2_s, K_s^2) - f^2(s, Y^2_s, Z^2_s, K_s^2)$.
For each $t \in [0,T]$, we then have
\begin{equation}\label{A26}
e^{\beta  t} (Y^1_s - Y^2_s)   ^2 \leq   \eta \,{\mathbb E}[ \int_t^T e^{\beta  s} \bar f(s)^2  ds \mid 
{\cal G}_t ] \;\; \text{ \rm a .s.}\, 
\end{equation}
Moreover, 
$\|\bar Y \|_\beta^2 \leq T  \eta
\|\bar f \|_\beta^2,$ and if $\eta < \frac{1}{C^2}$, we then have 
$\|\bar Z \|_\beta^2 + \|\bar K \|_{\lambda,\beta}^2
\leq \frac{\eta}{1 - \eta C^2}  \|\bar f \|_\beta^2.$
\end{lemma}
The proof is similar to the one given for DRSBDEs in the same framework with default 

(see the proof of Proposition 6.1 in the Appendix in \cite{DQS5}), and left to the reader.

\end{document}